\documentclass[letterpaper, 10 pt, conference]{ieeeconf}  

\IEEEoverridecommandlockouts                              
                                                          
\usepackage{graphics} 
\usepackage{balance}
\usepackage{setspace}
\usepackage{pdfsync}
\usepackage{epstopdf}
\usepackage{subfigure}
\usepackage{amsmath, amsfonts, amssymb,amscd}
\usepackage[ruled,vlined]{algorithm2e} 

\newtheorem{lemma}{Lemma}

\usepackage{epsfig} 
\usepackage{psfrag}
\usepackage{amsmath} 

\usepackage{color}

\usepackage{mathtools}
 \usepackage{setspace}
 \usepackage{epsfig}
 \usepackage[ruled,vlined]{algorithm2e} 

\usepackage{cite}

\usepackage{color}

\title{\bf  Invariant Properties of Linear-Iterative Distributed Averaging Algorithms and  Application to Error Detection}
 
\author{\authorblockN{Christoforos N. Hadjicostis and Alejandro D. {Dom\'{i}nguez-Garc\'{i}a}}
\thanks{C. N. Hadjicostis is with the ECE Department  at the University of Cyprus, Nicosia, Cyprus, and also with the ECE Department   at the University of Illinois at Urbana-Champaign, Urbana, IL 61801, USA. E-mail: {\tt\small hadjicostis.christoforos@ucy.ac.cy}.}
\thanks{Alejandro D. {Dom\'{i}nguez-Garc\'{i}a} is with the ECE Department  at the University of Illinois at Urbana-Champaign, Urbana, IL 61801, USA. E-mail: {\tt\small aledan@illinois.edu}.}
}

\begin{document}

\maketitle

\begin{abstract}
We consider the problem of average consensus in a distributed system comprising a set of nodes that can exchange information among themselves.  We focus on a class of algorithms for solving such a problem  whereby each node   maintains a state  and updates it iteratively as a linear combination of the states maintained by its in-neighbors, i.e., nodes from which it receives information directly. Averaging algorithms within   this class   can be thought of as discrete-time  linear  time-varying systems without external driving inputs and whose state matrix is column  stochastic.  As a result, the algorithms exhibit a   global invariance property in that the sum of the state variables remains constant at all times. In this paper, we report on  another invariance property  for the aforementioned  class of averaging algorithms. This property is local to each node and reflects the conservation of  certain quantities capturing an aggregate of all the values received by a node from its in-neighbors and all the values sent by said node to its out-neighbors (i.e., nodes to which it sends information directly) throughout the execution of the averaging algorithm. We show how this newly-discovered invariant can be leveraged  for detecting errors while executing the  averaging algorithm.
%
%
%
\end{abstract}

\section{Introduction}
 We consider distributed systems that consist of a set of nodes that can exchange information among themselves. In such systems, it is often necessary for all or some of the nodes to calculate a function of certain parameters, with  each of these  possessed by an individual node \cite{1984:Tsitsiklis,2018:BOOK}. For example, when all nodes calculate the average of these parameters, they are said to reach average consensus. Because of  its usage in numerous distributed control, computing, and communication  applications,   average consensus algorithms for distributed systems have been researched extensively over the last few years (see, e.g., \cite{1984:Tsitsiklis,2018:BOOK,1996:Lynch, 2004:Murray,2008:Cortes}). 
 
 An important class of such distributed averaging algorithms  rely on each node maintaining a state  that is updated iteratively as a linear combination of the states maintained by in-neighbors, i.e., nodes from which it receives information directly    (see, e.g.,  \cite{1984:Tsitsiklis,2004:Murray,2018:BOOK} and the references therein). Such class includes, e.g., gossip algorithms \cite{boyd2006randomized},   the push-sum algorithm and its variants \cite{2003:Kempe,2010:Benezit},   the ratio-consensus algorithm and its variants \cite{2010:christoforos,dominguez2012resilient,nitinTAC2016}, and several others.  As it turns out, one can think of the linear iterations on which the aforementioned class of algorithms rely as discrete-time linear time-varying  systems without external driving inputs and whose state matrix is column  stochastic. Thus, many results in this area can be leveraged when analyzing the behavior of such algorithms (the book by Seneta \cite{Se:06}, originally published in 1973,  is a must-read reference). 
 
 Because of the aforementioned column-stochasticity property, the linear-iterative average consensus algorithms of interest in this paper exhibit the following (well-known) invariance property:  if one assumes that the computations performed by each node in updating its state are error free,  the sum of the states for all nodes in the system must remain constant and equal to the sum of the initial values. However, this invariance property might not hold in general  if such computations contain errors. Thus, one could leverage this property to detect the presence of computation errors.
 
An issue with the invariance property discussed above is that in order to check it, one needs to have access to all the states of  the system. This implies that if this property were to be used to implement an error detection scheme, it would be necessary to have a processor with access to the states of all individual nodes, i.e., such error detection would be essentially centralized.   In this paper, we report on  another invariance property  for the aforementioned  class of   averaging algorithms   that we have discovered. To the best of our knowledge, the existence of such property has not been reported in the literature before. As it turns out, the newly-discovered invariant is local to each node and reflects a conservation property involving certain quantities capturing an aggregate of all the values received by a node from its in-neighbors   and all the values sent by said node to its out-neighbors (i.e., nodes to which it sends information directly) during the execution of the averaging algorithm.

The local nature of the newly-discovered invariant makes it ideal for implementing distributed schemes for detecting errors in the computations performed by the nodes of a distributed system executing  any averaging algorithm within the class of interest. In this paper,  we propose one such error detection scheme that allows each node  in the system  to check whether or not its in-neighbors  have performed their updates correctly. In particular, each node uses information that it receives periodically from {\em two-hop}  in-neighbors, i.e., the in-neighbors of its in-neighbors. Using two-hop information is also leveraged in \cite{yuan2019resilient, yuan2021resilient, yuan2021secure} in the context of detection of faulty/malicious nodes in a distributed system. However, unlike our proposed error detection scheme, all these other  schemes require two-hop information at each iteration.

\section{Communication Topology Model} \label{SECTopologyModel}

Consider a distributed system comprising $N$ nodes, denoted by $v_i,~i=1,2,\dots,N$. Assume that the nodes can potentially exchange information among themselves. We capture such exchange of information  by  a strongly connected directed graph $\mathcal{G}=\{\mathcal{V},\mathcal{E} \}$, with $\mathcal{V} =  \{v_1, v_2, \dots, v_N\}$ and $\mathcal{E} \subseteq \mathcal{V} \times \mathcal{V} - \{ (v_j,v_j)$ $|$ $v_j \in \mathcal{V} \}$ such that $(v_j, v_i) \in \mathcal{E}$ if node $v_j$ may receive information from node $v_i$.\footnote{We say that a directed graph $\mathcal{G}= \{\mathcal{V}, \mathcal{E}\}$ is  \text{strongly connected} if for each pair of nodes $v_j, v_i \in \mathcal{V}$, $v_j \neq v_i$, there exists a directed \textit{path} from $v_i$ to $v_j$, i.e., we can find a sequence of nodes $v_i =: v_{l_0},v_{l_1}, \dots, v_{l_t} := v_j$ such that $(v_{l_{\tau+1}},v_{l_{\tau}}) \in \mathcal{E}$ for $ \tau = 0, 1, \dots , t-1$.}  Note that the exchange of information between pairs of nodes may be \textit{asymmetric}, i.e., $v_j$ may receive information from node $v_i$ but not vice-versa; in such case, we have that  $(v_j, v_i) \in \mathcal{E}$ but $(v_i, v_j) \notin \mathcal{E}$.
 
Nodes that can potentially send information to node $v_j$ directly are said to be its in-neighbors and belong to 
the set $\mathcal{N}_j^- := \{v_i \in \mathcal{V} \; | \; (v_j, v_i) \in \mathcal{E} \}$, which is referred to as the in-neighborhood of node $v_j$. Similarly, nodes that can potentially receive information from node $v_j$ directly are said to be its out-neighbors and belong to the set $\mathcal{N}_j^+ := \{v_l \in \mathcal{V} \; | \; (v_l, v_j) \in \mathcal{E} \}$, which  is referred to as the out-neighborhood of node $v_j$. The cardinality of  $\mathcal{N}_j^-$, which we denote by $D_j^-$, is referred to as the \textit{in-degree} of $v_j$, whereas the cardinality of  $\mathcal{N}_j^+$, which we denote by $D_j^+$, is referred to as the \textit{out-degree} of $v_j$.

\section{Distributed Averaging via Linear Iterations} \label{sec:algorithms}

Consider an $N$-node distributed system whose communication topology is described by a strongly connected directed graph $\mathcal{G}=\{\mathcal{V},\mathcal{E} \}$. Assume that each node $v_j$ possesses a value $V_j$ and the objective for all the nodes is to compute the average of the $V_j$'s, i.e., 
\begin{align}
\overline{V}:=\frac{\sum_{l=1}^N V_l}{N} \; . \label{DEFv}
\end{align} 
Next, we describe a family of iterative algorithms that allow each node $v_j \in \mathcal{V}$ to compute $\overline{V}$; such algorithms are distributed in the sense that they conform to the constraints imposed  on the exchange of information among nodes as captured by $\mathcal{G}=\{\mathcal{V},\mathcal{E} \}$.

At discrete time instants indexed by $k=0,1,2,\dots,$ each node $v_j$ sends some information to a subset of its out-neighbors, which we denote by $\mathcal{N}_j^+[k]\subseteq \mathcal{N}_j^+$. This implies that at each time instant $k$, node $v_j$ only receives information from a subset of  nodes in $\mathcal{N}_j^-$; we denote such set by  $\mathcal{N}_j^-[k]\subseteq \mathcal{N}_j^-$. Each node $v_j$ maintains two state variables,  $y_j[k]$ and $z_j[k]$, and uses them to obtain $\overline{V}$ as follows.  Let  $x_j[k]=\big [y_j[k] , z_j[k]\big]^\top$; then,  each node $v_j$ performs the following  operations:
\begin{itemize}
\item[\bf O1.] Initially, it sets $x_j[0]=[V_j,1]^\top$. 
\item[\bf O2.] At each $k=0,1,2,\dots$, it sends the value $w_{lj}[k]x_j[k]$ to each  $v_l \in \mathcal{N}_j^+[k]$, where $w_{lj}[k]$ is some time-varying weight whose choice is described later.
\item[\bf O3.] At each $k=0,1,2,\dots$, it receives the value $w_{ji}[k] x_i[k]$ from each $v_i \in \mathcal{N}_j^-[k]$. 
\item[\bf O4.] It updates the value of $x_j[k]$ at each $k=0,1,2,\dots$, as follows:
\begin{align}
x_j[k+1]&=w_{jj}[k] x_j[k]+\sum_{v_i \in \mathcal{N}_j^-[k]  } w_{ji}[k]x_i[k],   \label{alg_x} 
\end{align}
where $w_{jj}[k]$ is a time-varying weight  whose choice is described later.
\item[\bf O5.] At each $k=0,1,2,\dots$, it also calculates  
\begin{align}
r_j[k] := \frac{y_j[k]}{z_j[k]}. \label{ratio_def}
\end{align}
\end{itemize}
Below, we will argue that proper choice of  the $w_{ji}[k]$'s together with some standard assumption regarding how often nodes receive information from their in-neighbors will result in  the $r_j[k]$'s asymptotically converging to $\overline{V}$.

\subsection{Weight Choice} Let $y[k]=\big[y_1[k],y_2[k],\dots,y_N[k]\big]^\top$ and $z[k]=\big[z_1[k],z_2[k],\dots,z_N[k]\big]^\top$, and define
\begin{align} 
X[k]&=\big [y[k], z[k] \big ] \nonumber \\
&=\big[x_1[k],x_2[k],\dots,x_N[k]\big]^\top \in \mathbb{R}^{N \times 2};
\end{align}  then we can rewrite \eqref{alg_x} in matrix form as follows:
\begin{align}
X[k+1]=W[k]X[k], \label{alg_x_matrix} 
\end{align}
where $W[k] =\big [w_{l,j}[k]\big]\in \mathbb{R}^{ N \times N}$ is referred to as the weight matrix at time instant $k$ with $w_{l,j}[k]=w_{lj}[k]$ if $v_l \in \mathcal{N}_j^+[k] \cup \{v_j \}$, and $w_{l,j}[k]=0$ otherwise. Now, choose the entries of the weight matrix $W[k]$ so as to satisfy the following properties: 
\begin{itemize}
\item[\bf W1.] $w_{l,j}[k] =0,~v_l \notin \mathcal{N}_j^+[k] \cup \{v_j \}$,
\item[\bf W2.]  $w_{l,j}[k] \in (\varepsilon,1-\varepsilon)$, where $\varepsilon>0$ is small and bounded away from zero, and
\item[\bf W3.]  $w_{j,j}[k]=1-\sum_{l \in  \mathcal{N}_j^+[k]} w_{l,j}[k], $ such that $w_{j,j}[k]>\varepsilon$, 
\end{itemize}
for all $v_j \in \mathcal{V}$.  

Such choice of weights, which results in    $W[k]$ being a column stochastic matrix,  is appealing because each node $v_j$ can easily implement it independently of all other nodes as illustrated in the following two special cases.
 \subsubsection{Ratio-Consensus Algorithm}  Consider the particular case in which each node $v_j \in \mathcal{V}$ sends information at each $k=0,1,2\dots$, to all its out-neighbors, i.e., $\mathcal{N}_j^+[k]=\mathcal{N}_j^+$ for all $k$, which implies that  $\mathcal{N}_j^-[k]=\mathcal{N}_j^-$ for all $k$. Assume that each node $v_j \in \mathcal{V}$ sets $w_{lj}[k]=\frac{1}{1+D_j^+}$ for all $v_l \in \mathcal{N}_j^+[k] \cup \{v_j\}=\mathcal{N}_j^+  \cup \{v_j\}$, where $D_j^+=\big | \mathcal{N}_j^+  \big |$. Then, \eqref{alg_x} reduces to
\begin{align}
x_j[k+1]&=\frac{1}{1+D_j^+} x_j[k]+\sum_{v_i \in \mathcal{N}_j^-  } \frac{1}{1+D_i^+}x_i[k],   \label{alg_x_LTI_1} 
\end{align}
with $x_j[0]=[V_j,1]^\top$ for all $v_j \in \mathcal{V}$. The specific instance of the algorithm described by Operations~O1--O5 that results from this choice of weights was   proposed in \cite{2010:christoforos,dominguez2012resilient}, and is referred to as the ratio-consensus algorithm.

\subsubsection{Generalized Push-Sum} Consider the general iteration described earlier and let $D_j^{+}[k]$ denote the number of out-neighbors of node $v_j$ to which it sends information at instant $k$, i.e., $D_j^{+}[k]=\big |\mathcal{N}_j^+[k]\big |$; then, node $v_j$ can set $w_{lj}[k]=\frac{1}{1+D_j^+[k]},~v_l \in \mathcal{N}_j^+[k] \cup \{v_j\}$. In this case, \eqref{alg_x} reduces to
\begin{align}
x_j[k+1]&=\frac{1}{1+D_j^+[k]} x_j[k]+\sum_{v_i \in \mathcal{N}_j^-[k]  } \frac{1}{1+D_i^+[k]}x_i[k],   \label{alg_x_LTV_1} 
\end{align}
with $x_j[0]=[V_j,1]^\top$ for all $v_j \in \mathcal{V}$. The specific instance of the algorithm described by Operations~O1--O5 that results from this choice of weights was  proposed in \cite{2010:Benezit}, and it is a generalized version of the so-called push-sum algorithm, which was first proposed in \cite{2003:Kempe}.

\subsection{Convergence Analysis}

Since each node $v_j \in \mathcal{V}$ may not receive information from  all its in-neighbors  at each time instant, we can describe the exchange of information among nodes at instant $k$ by a directed graph $\mathcal{G}[k]=\big \{\mathcal{V},\mathcal{E}[k] \big \},$ where $(v_l,v_j) \in \mathcal{E}[k]$ if $v_l \in \mathcal{N}_j^+[k]$. In the remainder of the paper, we make the following assumption  about the collection of graphs $\mathcal{G}[k]=\big \{\mathcal{V},\mathcal{E}[k] \big \},~k=0,1,2\dots,$ which is standard in the literature of average consensus (see, e.g., \cite{2018:BOOK}).

{\bf Assumption~A1.} There exists a finite $K$ such that 
\begin{align}
\mathcal{G}_K(\tau)
:= \big \{ \mathcal{V},  \cup_{t=0}^{K-1} \mathcal{E}[\tau K+t]  \big \}, \quad \tau=0,1,2,\dots,
\end{align}
is strongly connected. 

Now by using \eqref{alg_x_matrix}, we have that
 \begin{align}
 X[k]=\underbrace{W[k-1]W[k-2]\cdots W[0]}_{=:T[k]}X[0];
 \end{align}
 thus, convergence of $X[k]$ as $k \rightarrow \infty$ is governed by the column stochastic matrix $T[k]$. Then, under Assumption~A1, as $k$ goes to infinity,  we have that
 \begin{align}
 T[k]=\pi[k] \textbf{1}_N^\top, \label{weak_ergodicity}
 \end{align}
 where $\pi[k]=\big[\pi_1[k],\pi_2[k],\dots,\pi_N[k]\big]^\top$ with $0<\pi_i[k]<1$ for all $i=1,2,\dots,N,$ such that $\sum_{i=1}^N \pi_i[k]=1$, and   $ \textbf{1}_N$ denotes the $N$-dimensional all-ones vector. In words, the columns of $T[k]$ will equalize asymptotically, but they do not necessarily converge, i.e., in general they will still depend on $k$. The literature of stochastic matrices refers to this convergence result as weak ergodicity of the sequence of $W[k]$'s,    whereas if  $\pi[k]$ converges to a limit, it is said that strong ergodicity of the sequence of $W[k]$'s  obtains \cite{Se:06}. Finally, by using \eqref{weak_ergodicity}, as $k$ goes to infinity,  we have that
 \begin{align}
 y_j[k] &=\pi_j[k] \Big (\sum_{l=1}^N y_l[0] \Big), \\
z_j[k] &=\pi_j[k] \Big (\sum_{l=1}^N z_l[0] \Big). 
 \end{align}
Thus,   the ratio $r_j[k] = y_j[k]/z_j[k]$ asymptotically converges to the average of the initial values,  i.e.,
 \begin{align}
 \lim_{k\rightarrow\infty} r_j[k]& =\frac{y_j[k]}{z_j[k]} \nonumber \\
 & = \displaystyle \frac{\sum_{l=1}^N y_l[0]}{\sum_{l=1}^N z_l[0]} \nonumber \\ 
 &= \overline{V}, \quad  \forall v_j \in \mathcal{V}, \label{EQratiolim}
 \end{align}
 despite the fact that $y_j[k]$ and $z_j[k]$ do not converge in general. 
 
 \subsubsection*{Balanced Weights}
  Under Assumption~A1, there is one special case that guarantees convergence of $T[k]$ to a limit. Namely, when the matrices $W[k],~k=0,1,2,\dots,$ are doubly stochastic, i.e., in addition to satisfying Properties~W1--W3, the entries of $W[k]$ are such that  $\sum_{i=1}^N w_{j,i}[k]=1$ for all $v_j \in \mathcal{V}$; in such case, we say the weights are balanced.  Then, it can be shown that
\begin{align}
\lim_{k \rightarrow \infty} T[k]=\frac{1}{N}\textbf{1}_N\textbf{1}_N^\top;
\end{align}
thus, 
\begin{align}
\lim_{k \rightarrow \infty} y_j[k]  &=\frac{\sum_{l=1}^N y_l[0] }{N}, \\
&= \frac{\sum_{l=1}^N V_l }{N}=\overline{V}, \\
\lim_{k \rightarrow \infty} z_j[k] &=1. 
\end{align}
Clearly, in such case, each node $v_j$ only needs to maintain one variable, namely $y_j[k]$,  to be able to obtain the average. However, it is not easy  to obtain, in a distributed manner, a weight matrix $W[k]$ that is doubly stochastic and conforms to  a general directed graph $\mathcal{G}[k] =\big \{\mathcal{V},\mathcal{E}[k] \big \}$ (see, for example, \cite{2013TACmatrixbalaning}), unless the exchange of information between pairs of nodes is \textit{symmetric}, i.e., if at instant $k$ node $v_j$ receives information from node $v_i$ then node $v_i$ receives information from node $v_j$. In such case, we have that  $(v_j, v_i) \in \mathcal{E}[k]$ if and only if $(v_i, v_j) \in \mathcal{E}[k]$, from where it follows that $\mathcal{N}_j^+[k]=\mathcal{N}_j^-[k] =:\mathcal{N}_j[k]$. Then, each $v_j \in \mathcal{V}$ can set  $w_{lj}[k]=\frac{1}{N},~v_l \in \mathcal{N}_j[k]$, and $w_{jj}[k]=1-\frac{D_j[k]}{N}$, where $D_j[k]=\big |\mathcal{N}_j[k]\big |$; otherwise, if $v_l \notin \mathcal{N}_j[k] \cup \{ v_j \}$, node $v_j$ sets $w_{lj}[k]=0$. Other choices also exist.

 \section{Invariant Properties of  Linear-Iterative Distributed Averaging Algorithms} \label{sec:invariants}

In this section, we first review a  well-known invariant property of the class of linear-iterative distributed averaging algorithms discussed in Section~\ref{sec:algorithms}. Then, we  introduce another invariance property that we have discovered. Specifically, we show that   there is a quantity associated to each node   that must remain invariant throughout the execution of the algorithm if there are no errors in the computations performed by each node in updating its state.  If said quantity  varies from the value it is supposed to have, it could be an indicator that there are errors in the computations that  a particular node  is performing; such application of the newly-discovered invariant is discussed in  Section~\ref{sec:error_detection_app}. 

 \subsection{Global Invariant}
 
As discussed earlier, and shown in \eqref{alg_x_matrix}, the averaging algorithms of interest are undriven discrete-time linear  time-varying systems whose state matrix is column  stochastic. As a result,  the sum of the state variables remains constant at all times; such  invariance property is formally established next.

 \begin{lemma}
 Consider \eqref{alg_x} with the  $w_{ji}[k]$'s  chosen so that the weight matrix $W[k]=\big [w_{j,i}[k]\big] \in \mathbb{R}^{N \times N}$ in  \eqref{alg_x_matrix} satisfies Properties~W1--W3. Then,
 \begin{align}
 \sum_{j=1}^N x_j[k]=  \sum_{j=1}^N x_j[0]
 \end{align}
 for all $k=0,1,2,\dots$.
 \end{lemma}
 
 \begin{proof}
Since $W[k]$ is column stochastic, we have that $\textbf{1}_N^\top W[k]=\textbf{1}_N^\top$ for all $k=0,1,2,\dots$; then, by using \eqref{alg_x_matrix}, it follows that
 \begin{align}
 \textbf{1}_N^\top X[k+1] &=  \textbf{1}_N^\top W[k]X[k] \nonumber \\
 &=  \textbf{1}_N^\top  X[k],
 \end{align}
 for $k=0,1,2,\dots$;  thus, $\sum_{j=1}^N x_j[k] = \sum_{j=1}^N x_j[0]$ for all $k=0,1,2,\dots$.
 \end{proof}
 
Note that because  the  invariance property  is global,  one would need to have access to all the states of the system in order to check it. In fact, because of the way \eqref{alg_x} is initialized, if such global information were readily available, one would immediately know the value of the sum of the $V_j$'s and the number of nodes in the system, and thus it would be straightforward to compute $\overline{V}$.

 \subsection{Local Invariant}

Next, we introduce   another invariance property  satisfied by the averaging algorithms considered in this paper. This invariance property is local to each node and reflects a conservation property involving certain quantities that amount to aggregates of all the values received by a node from its in-neighbors and all the values sent by said node to its out-neighbors throughout the execution of the averaging algorithm. The property is formally stated as follows. 

\begin{lemma} 
Consider \eqref{alg_x} with the  $w_{ji}[k]$'s  chosen so that the weight matrix $W[k]=\big [w_{j,i}[k]\big] \in \mathbb{R}^{N \times N}$ in  \eqref{alg_x_matrix} satisfies Properties~W1--W3. For each $v_j \in \mathcal{V}$,  and each $v_l \in \mathcal{N}_j^+$, define
\begin{align}
\sigma_{lj}[k+1]=\begin{cases} 
\sigma_{lj}[k]+w_{lj}[k]x_j[k], &  \text{if} \quad v_l \in \mathcal{N}_j^+[k], \\ 
\sigma_{lj}[k], &   \text{if} \quad v_l \in \mathcal{N}_j^+ \backslash   \mathcal{N}_j^+[k],
\end{cases}
\end{align}
for $k=0,1,2,\dots$, 
with $\sigma_{lj}[0]=0$. Then, 
 \begin{align}
 \eta_j[k]:=&~x_j[k]+\sum_{v_l \in \mathcal{N}_j^+}\sigma_{lj}[k]-\sum_{v_i \in \mathcal{N}_j^-  } \sigma_{ji}[k] \nonumber \\
= &~x_j[0],   \label{alg_x_LTV_1_inv} 
 \end{align}
 for all $k =0,1, 2 \dots$.
 \end{lemma}
 
 \begin{proof}
 By induction. For $k=0$, we have that
\begin{align}
 \eta_j[0] =&~ x_j[0]+ \underbrace{\sum_{v_l \in \mathcal{N}_j^+ } \sigma_{lj}[0]}_{=0}- \underbrace{\sum_{v_i \in \mathcal{N}_j^- } \sigma_{ji}[0]}_{=0} \nonumber \\
 =&~   x_j[0].
 \end{align}

 Assume that \eqref{alg_x_LTV_1_inv} holds for some $k \geq 0$.  Then, for $k+1$, we have that
\begin{align}
 \eta_j[k+1] =&~ x_j[k+1]+\sum_{v_l \in \mathcal{N}_j^+ } \sigma_{lj}[k+1]-\sum_{v_i \in \mathcal{N}_j^- } \sigma_{ji}[k+1] \nonumber \\
 =&~   w_{jj}[k]x_j[k]+\sum_{v_i \in \mathcal{N}_j^-[k]  } w_{ji}[k] x_i[k]    \nonumber \\
  &+   \sum_{v_l \in \mathcal{N}_j^+  }    \sigma_{lj}[k]  + \sum_{v_l \in \mathcal{N}_j^+[k] }   w_{lj}[k] x_j[k]    \nonumber \\
 &-  \sum_{v_i \in \mathcal{N}_j^-  }  \sigma_{ji}[k] -\sum_{v_i \in \mathcal{N}_j^-[k] }  w_{ji}[k] x_i[k]   \nonumber \\
=&~ \underbrace{\left( w_{jj}[k] +\sum_{v_l \in \mathcal{N}_j^+[k] } w_{lj}[k]  \right)}_{=1} x_j[k]  \nonumber \\
&  +    \sum_{v_l \in \mathcal{N}_j^+ }  \sigma_{lj}[k] -\sum_{v_i \in \mathcal{N}_j^- }  \sigma_i[k]  \nonumber \\
 =&~\eta_j[k];
 \end{align}
 thus, since  by the inductive assumption we have that $\eta_j[k]=x_j[0]$, it follows that $\eta_j[k+1]=x_j[0]$.   \end{proof}
 
 For the special case in \eqref{alg_x_LTI_1}, i.e., the so-called ratio-consensus algorithm, the invariant quantity associated with each $v_j \in \mathcal{V}$ is slightly simpler in that we do not need separate $\sigma_{lj} [k]$'s for each $v_l \in \mathcal{N}_j^+$; this is captured by the  result in the following lemma.

\begin{lemma}
Consider the linear iteration in \eqref{alg_x_LTI_1}. For  each $v_j \in \mathcal{V}$, define
\begin{align}
\sigma_j[k+1]&= \sigma_j[k] + \frac{1}{1+D_j^+}x_j[k], \quad k =1, 2,\dots, \label{runningsumratio}
\end{align}
with $\sigma_{j}[0]=0$.  Then, 
 \begin{align}
 \eta_j[k]:=&~x_j[k]+D_j^+\sigma_j[k]-\sum_{v_i \in \mathcal{N}_j^- } \sigma_i[k] \nonumber \\
= &~x_j[0],   \label{alg_x_LTI_1_inv} 
 \end{align}
 for all $k =0,1,  \dots$.
 \end{lemma}
 
 \begin{proof}
By induction.  For $k=0$, we have that
\begin{align}
 \eta_j[0] =&~ x_j[0]+\underbrace{D_j^+\sigma_j[0]}_{=0}-\underbrace{\sum_{v_i \in \mathcal{N}_j^- } \sigma_i[0] }_{=0}\nonumber \\ 
 =&~x_j[0].
 \end{align} 

Assume \eqref{alg_x_LTI_1_inv} holds for some $k \geq 0$. Then, for $k+1$, we have that
\begin{align}
 \eta_j[k+1] =&~ \underbrace{ \frac{1}{1+D_j^+} x_j[k]+\sum_{v_i \in \mathcal{N}_j^-  } \frac{1}{1+D_i^+}x_i[k]  }_{=x_j[k+1]}\nonumber \\
 &+D_j^+  \underbrace{\left (\sigma_j[k]+\frac{1}{1+D_j^+} x_j[k] \right)  }_{=\sigma_j[k+1]} \nonumber \\
 &-\sum_{v_i \in \mathcal{N}_j^- } \underbrace{ \left  (\sigma_i[k]+\frac{1}{1+D_i^+} x_i[k] \right) }_{=\sigma_i[k+1]}\nonumber \\
=&~ x_j[k]   +    D_j^+ \sigma_j[k] -\sum_{v_i \in \mathcal{N}_j^- }  \sigma_i[k]  \\
 =&~\eta_j[k];
 \end{align}
 thus, since  by the inductive assumption we have that $\eta_j[k]=x_j[0]$, it follows that $\eta_j[k+1]=x_j[0]$ .
 \end{proof}

 \section{Error Detection} \label{sec:error_detection_app}

The possible presence of malfunctioning nodes is one of the biggest headaches in distributed averaging schemes. It is easy to see that a faulty node can single handedly cause a deviation of the consensus value from the average $\overline{V}$. For example, if node $v_i$ is ``stubborn" in the sense that it does not follow the update in \eqref{alg_x} but maintains its value at all time instants, i.e., $x_i[k] = V_i$ for all $k$, then the final consensus value for {\em all} nodes will be $V_i$ \cite{pasqualetti2007distributed,sundaram2011distributed}. Thus, in order to ensure that averaging algorithms work properly, it is desirable to have a mechanism built-in the algorithm to detect such errors during execution, and  correct them. In this paper, we focus on the detection part.

Next, we will briefly discuss some existing approaches, based on modular/time redundancy or parity checking techniques, for detecting/correcting such computational errors in discrete-time linear time-varying systems whose dynamics is similar to those of the averaging algorithms in the class discussed in Section~\ref{sec:algorithms}. Then, we will introduce a novel approach, which we refer to as \textit{any-time consistency checking}, that is related to parity checking, and discuss its key properties that make it appealing for detecting computational errors in faulty nodes that are executing distributed averaging algorithms within the class of iterest. Finally, we illustrate how any-time consistency checking can be leveraged against error detection using  the invariants identified in Section~\ref{sec:invariants}.

\subsection{Existing Techniques for Error Detection/Correction}

Error detection and correction techniques typically rely on introducing redundancy in the computation by performing the same computation, either multiple times (time redundancy) or over multiple system replicas (modular redundancy) \cite{koren2020fault}. For example, in a modular redundancy scheme, one would utilize several replicas of the given system, initialized identically and driven by the same inputs, and would compare the (ideally identical) results provided by them; in case of a disagreement, one would choose the most likely result (e.g., the one agreed upon by the majority of the system replicas) \cite{koren2020fault}. In dynamic settings where a computation is performed over several time steps, one can choose to perform checks {\em concurrently} (i.e., at the end of each time step) or {\em non-concurrently} (e.g., periodically). Modular (time) redundancy schemes impose hardware (time) overhead; concurrent checking can also impose significant delays, whereas non-concurrent checking has to deal with error propagation over time (since an error is not caught immediately and it is allowed to propagate, at least until the next check is performed). 

Parity check techniques offer an alternative to time/modular redudancy schemes, but have to identify inherent invariant properties among the given system's variables (if any). The validity of an invariant is checked during execution and is associated with a parity check \cite{patton1991review}; when the later is violated, an error is detected and (if possible) the outcomes of the parity checks can be leveraged towards error correction. In computational systems, when inherent invariant properties cannot be identified (or are not sufficient to capture all errors of interest), one can attempt to introduce such invariant properties by using encoding techniques; resulting approaches are referred to as information redundancy \cite{koren2020fault}. For linear dynamic systems, information redundancy has led to approaches that introduce and update at each time step a small number of additional variables, whose values can be used to perform concurrent or non-concurrent detection/identification of computational errors that may have been introduced in the system variables (or even in the additional variables themselves) \cite{hadjicostis2003nonconcurrent}.

The implementation of the above approaches can be challenging in distributed settings like the one in \eqref{alg_x}, because of constraints on the information that is available to the added redundancy or the checker. One such example appears in \cite{yuan2019resilient} and \cite{yuan2021secure}, both of  which focus on distributed systems with symmetric exchange of information between pairs of nodes,  and aim at detecting errors due to malicious activity in one or more nodes. The authors proposed a concurrent checking scheme where, at each time step $k$, each node $v_j$ checks each neighbor~$v_i$ assuming that node $v_j$ has access to all its ``inputs" (values provided to $v_i$ by its neighbors, which are generally two-hop neighbors of node $v_j$). Distributivity constraints aside, this approach amounts to multiple time redundancy schemes in which checks are performed concurrently. Also in the context of  \eqref{alg_x} for the case when the exchange of information between pairs of node is symmetric, the approach in \cite{hadjicostis2023identification} explored the application of the non-concurrent checking schemes in \cite{hadjicostis2003nonconcurrent} towards distributed detection and identification of malicious nodes.

 \subsection{Any-Time Consistency Checking}

Considering the linear iteration in \eqref{alg_x}, we are interested in {\em any-time} consistency checks (ATC's) that resemble parity checks (in the sense that they can be used to detect computational errors) but can be performed at any point in time (i.e., they do not need to be checked at each time step). Moreover,  such ATC's must have the following two   properties:
\begin{itemize}
\item[\bf P1.] If all ATC's at time step $k_0$ are valid, then if the nodes continue the iteration in \eqref{alg_x} for $k=k_0, k_0+1, k_0+2, \dots,$ {\em without} any subsequent errors, they will converge to the correct average of their initial values, $\overline{V}$.
\item[\bf P2.] If one or more ATC's at iteration $k_0$ are invalid, then there have been errors by at least one node at one or more time steps before $k_0$. 
\end{itemize}
Note that the above two properties leave out the possibility for the checks at time step $k_0$ to be valid even though there have been errors within the time window $0, 1, \dots,  k_0-1$, as long as this behavior does {\em not} affect the outcome of the distributed computation. In other words, the nodes will still converge to $\overline{V}$ if they continue the computation without further errors. Next we describe how to utilize the invariant properties established in Section~\ref{sec:invariants} to build any-time checking schemes with the aforementioned desired properties.

\subsection{Invariant-Based Error Detection}

In the remainder, we focus on a special case of the linear iterative scheme in \eqref{alg_x}; namely, the ratio-consensus algorithm described in Section~\ref{sec:algorithms}. However, we believe  the results can be  extended to the general case; we plan to pursue such  extension in future work.

We are given an $N$-node distributed system whose communication topology is described by a strongly connected directed graph $\mathcal{G}=\{\mathcal{V},\mathcal{E} \}$; the objective of the nodes is to compute the average $\overline{V}$ in \eqref{DEFv} utilizing the linear iteration in \eqref{alg_x_LTI_1}. We are interested in detecting computational errors that corrupt the entries of the $x_i[k]$ vector of a particular node~$v_i$ at one or more time steps. Apart from special cases (such as stalling or power failure which are easily detected), the effect of computational errors can be modeled via an additive error vector  $e_i[k]$, i.e.,   $x_i[k]$ changes  to $x_i[k] + e_i[k]$ at iteration step $k$. If no action is taken  to remedy the error, it will subsequently propagate to the out-neighbors of node $v_i$ (through the linear iteration in \eqref{alg_x_LTI_1}), then to their out-neighbors, and eventually to the whole network (since the corresponding graph is assumed to be strongly connected).

Instead of checking at each time step for the presence of errors (like $e_i[k]$ in the above discussion), we will build any-time checks based on the invariants in \eqref{alg_x_LTI_1_inv}. More specifically, the invariant associated with each node $v_i$ will be checked by each of its out-neighbors. Since the running sums in \eqref{runningsumratio} are needed to check the invariant, we modify a bit the implementation of the iteration in \eqref{alg_x_LTI_1}. More specifically, each node executes the iteration as 
\begin{align}
x_j[k+1]&=\frac{1}{1+D_j^+} x_j[k]+\sum_{v_i \in \mathcal{N}_j^-} (\sigma_i[k+1] - \sigma_i[k]),   \label{alg_x_LTV_1_running_sum}
\end{align}
which requires each node $v_j$ to maintain the following variables at each iteration $k$:
\begin{itemize}
\item[\bf V1.] The value $\sigma_j[k+1]$ (broadcasted to its out-neighbors at iteration $k$), in order to be able to update its own running sum (using \eqref{runningsumratio}) at the next iteration.
\item[\bf V2.]  The values $\sigma_i[k]$ for each in-neighbor $v_i \in \mathcal{N}^-_j$ (broadcasted by node $v_i$ at the previous time-step).
\end{itemize}

In addition, we make the following mild assumption.

{\bf Assumption~A2.} Each node $v_j$ knows the local topology around each of its in-neighbors, i.e., node $v_j$ is aware of the in-neighbors and out-neighbors of each $v_i \in \mathcal{N}^-_j$; in particular, node $v_j$ knows the out-degree $D^+_i$ of each in-neighbor $v_i \in \mathcal{N}_j^-$.

Now, let us consider an arbitrary node $v_j$ and an arbitrary in-neighbor of it denoted by $v_i$ (i.e., $v_i \in \mathcal{N}^-_j$). We use $v_{i'}$ to denote an arbitrary in-neighbor of node $v_i$ (i.e., $v_{i'} \in \mathcal{N}^-_i$) and $v_{l'}$ to denote an arbitrary out-neighbor of node $v_i$ (i.e., $v_{l'} \in \mathcal{N}^+_i$). Suppose that periodically (say once every $K$ iteration steps where $K$ is a design parameter), node $v_j$ receives additional information from the in-neighbors of node~$v_i$. Specifically, if we use $k_0$ to denote the index of one such iteration step, we assume that node $v_j$ receives the value of $\sigma_{i'}[k_0]$ from each $v_{i'} \in \mathcal{N}^-_i$. This could be easily achieved by having each node $v_{i'}$ transmit (every $K$ steps) its running sum $\sigma_{i'}[k_0]$ at a higher power, so that this running sum is received not only by its out-neighbors but also by the out-neighbors of its out-neighbors.\footnote{In the case of symmetric information exchange, these messages are equivalent to the two-hop information utilized in \cite{yuan2021secure}. Such transmissions are likely more expensive and undesirable,  and that is one of the reasons we propose a checking scheme that requires this information only infrequently.} [Alternatively, we can start with a dense directed graph, but then limit the immediate neighborhood of each node in order to enable the checking properties that we need (i.e., we  design the   network topology of the distributed system so that the information needed at each node $v_j$ is available).]

Now, since node $v_j$ also receives the values of $x_i[k]$ at each iteration, this node can check whether the invariant in \eqref{alg_x_LTI_1_inv} holds for node $v_i$, i.e., whether the following any-time consistency check (performed by node $v_j$ on node $v_i$ at time step $k_0$) 
$$
c_i^{(j)}[k_0] := \eta^{(j)}_i[k_0] - x_i[0]
$$
is zero or not, where $\eta^{(j)}_i[k_0]$ can be calculated by node $v_j$ via
$$
\eta^{(j)}_i[k_0] = x_i[k_0] + D_i^+ \sigma_i[k_0] - \sum_{v_{i'} \in \mathcal{N}_i^-  } \sigma_{i'}[k_0]
$$
as all needed information is available to it.  {Note that the value $x_i[0]$ for each in-neighbor $v_i \in \mathcal{N}^-_j$ can be obtained from $\sigma_i[1]$ transmitted by $v_i$ at time step $k=0$ as $x_i[0]=(1+D_i^+)\sigma_i[1]$.} It is also worth pointing out that other out-neighbors of node $v_i$ will also be checking node $v_i$ through ATC's like the above; in fact, all such ATC's will evaluate at the same value since they are based on identical information. For this reason, we can drop the superscript $j$ and simply refer to the ATC~$c_i$.

The following result establishes Property~P1 for the any-time consistency checks.

\begin{lemma}
\label{THEanyconvergence}
Suppose that all ATC's satisfy $c_i[k_0]=0$ for all $v_i \in \mathcal{V}$. If nodes continue the distributed averaging algorithm in \eqref{alg_x_LTI_1} for iterations $k_0$, $k_0+1$, \dots, without any errors, then all nodes will reach consensus to the average $\overline{V}$ in \eqref{DEFv}.
\end{lemma}

\begin{proof}
For all $v_i \in \mathcal{V}$, we have that $c_i[k_0]=0$ or, equivalently, that 
$$
x_i[k_0] + D^+_i \sigma_i[k_0] - \sum_{v_{i'} \in \mathcal{N}_i^-  } \sigma_{i'}[k_0] = x_i[0] \; .
$$
If we sum up the above equalities over all $v_i \in \mathcal{V}$, on the right of the equality we get $\sum_{v_i \in \mathcal{V}} x_i[0]$ and on the left we get
$$
\sum_{v_i \in \mathcal{V}} x_i[k_0] + \sum_{v_i  \in \mathcal{V}} D^+_i \sigma_i[k_0] - \sum_{v_i \in \mathcal{V}} \sum_{v_{i'} \in \mathcal{N}_i^-  } \sigma_{i'}[k_0] \; .
$$
The above quantity simplifies to $\sum_{v_i \in \mathcal{V}} x_i[k_0]$ (because each $\sigma_i$ appears exactly $D^+_i$ times in the double sum (with a negative sign). We conclude that
$$
\sum_{v_i \in \mathcal{V}} x_i[k_0] = \sum_{v_i \in \mathcal{V}} x_i[0] \; ;
$$
thus, the average consensus iteration that starts at iteration $k_0$ with values $x_j[k_0]$ for each $v_j \in \mathcal{V}$ converges to $\frac{\sum_{v_j \in \mathcal{V}} x_j[k_0]}{N}$, which is identical to the average $\overline{V}$ in \eqref{DEFv}.
\end{proof}

As mentioned earlier, if a computational error corrupts the value of $x_i[k]$ at time step $k$, the error will next propagate to its out-neighbors, by affecting the values $\sigma_i[k+1]$ that are transmitted to its out-neighbors (and thus the $x_j[k+2]$ values computed by each out-neighbor $v_j$, $v_j \in \mathcal{N}^+_i$); subsequently, the error will affect the values of nodes further away from $v_i$. It can be shown that, in such case, the error will show as a violation of the invariant for node $v_i$ for time steps after $k$, but will not affect the validity of the invariants of other nodes. This opens the door to the possibility of developing error identification and correction, however, we leave such developments for future work. 
 
It is possible for multiple errors to corrupt $x_i[k]$ over several time steps in a way that the invariant for node $v_i$ still holds at the end of the period (when ATC's are evaluated). In such case, the errors go undetected but, as established by Lemma~\ref{THEanyconvergence}, the nodes will converge to the correct average (at least as long as no more errors are introduced in the computation after the evaluation of ATC's). This also gives an opportunity to nodes that suffer temporal faults to make appropriate corrections.

\section{Concluding Remarks}
In this paper we have reported on a newly-discovered invariant property  for a class of linear-iterative algorithms used in distributed systems to solve the average consensus problem. Such property is local to each node in the distributed system, which makes it ideal for implementing distributed schemes for detecting errors during the execution of the algorithms. We have proposed one such error detection scheme  for one of the algorithms within the class considered in this paper; namely the ratio-consensus algorithm. 

In future work, we plan to extend  the proposed error detection schemes to include  all other algorithms within the class of interest. We also plan to leverage the information provided by the invariant to develop schemes for correcting the errors introduced in the computation and guarantee convergence. We will consider both the case when one wants to keep the values of all nodes (including the faulty ones) in the computation of the average, and the case when it is desirable to remove the values of faulty nodes from the average computation.

\bibliographystyle{IEEEtran}

\bibliography{bibliografia_consensusV6}

\end{document}